\documentclass[11pt,letterpaper]{article}
\usepackage{arxiv}

\usepackage{hyperref}
\usepackage{url}
\usepackage{amsmath,amssymb,amsthm,amsfonts}
\usepackage{algorithm}
\usepackage{graphicx}
\usepackage{subcaption}
\usepackage{enumitem}
\usepackage[noend]{algpseudocode}
\newtheorem{lemma}{Lemma}

\newtheorem{theorem}{Theorem}
\newtheorem{corollary}{Corollary}

\newcommand{\E}{\mathbb{E}}

\title{Consumer Autonomy or Illusion? Rethinking Consumer Agency in the Age of Algorithms}
\author{Pegah Nokhiz\\
  Cornell University, Cornell Tech\\
  \texttt{pegah.nokhiz@gmail.com} \\
  \And
   Aravinda Kanchana Ruwanpathirana \\
  National University of Singapore\\
\texttt{kanchana.ruwanpathirana@gmail.com}\\}

\begin{document}



\date{}
\maketitle
\begin{abstract}
Consumer agency in the digital age is increasingly constrained by systemic barriers and algorithmic manipulation, raising concerns about the authenticity of consumption choices. Nowadays, financial decisions are shaped by external pressures like obligatory consumption, algorithmic persuasion, and unstable work schedules that erode financial autonomy. Obligatory consumption (like hidden fees) is intensified by digital ecosystems. Algorithmic tactics like personalized recommendations lead to impulsive purchases. Unstable work schedules also undermine financial planning.

Thus, it is important to study how these factors impact consumption agency. To do so, we examine formal models grounded in discounted consumption with constraints that bound agency. We construct analytical scenarios in which consumers face obligatory payments, algorithm-influenced impulsive expenses, or unpredictable income due to temporal instability. Using this framework, we demonstrate that even rational, utility-maximizing agents can experience early financial ruin when agency is limited across structural, behavioral, or temporal dimensions and how diminished autonomy impacts long-term financial well-being. Our central argument is that consumer agency must be treated as a \emph{value} (not a given) requiring active cultivation, especially in digital ecosystems. The connection between our formal modeling and this argument allows us to indicate that limitations on agency (whether structural, behavioral, or temporal) can be rigorously linked to measurable risks like financial instability. This connection is also a basis for normative claims about consumption as a value, by anchoring them in a formally grounded analysis of consumer behavior.  As solutions, we study systemic interventions and consumer education to support value deliberation and informed choices. We formally demonstrate how these measures strengthen agency.

\end{abstract}

\section{Introduction}
The ability to make informed choices about spending and saving (autonomy in consumption) is crucial for promoting individual empowerment, financial and mental health, social equity, and sustainable economic practices \cite{gordon2024consumer,rubel2020algorithms,cheong2024transparency}.  However, consumer behavior is shaped not only by individual preferences but also by social norms and peer influences. In the digital economy, algorithms and personalized marketing further complicate these dynamics. Algorithms frequently generate filter bubbles, restricting access to diverse viewpoints. Through manipulative design strategies and persuasive advertising, they steer consumers toward impulsive buying decisions, thereby diminishing their sense of agency \cite{pariser2011filterbubble,mathur2019darkpatterns}. Systemic barriers also play a critical role in limiting agency. For example, attributing the purchase of unhealthy groceries to personal laziness (overemphasizing individual choice and placing individual blame) neglects broader structural issues related to consumers' self-determination, such as limited access to affordable nutritious food and the pervasive presence and promotion of unhealthy options \cite{ver2010access}. These dynamics raise critical concerns about the authenticity of consumer agency in the digital age: Are individuals today genuinely making independent decisions about their saving and spending, or is agency in consumption becoming increasingly illusory? Consider the following examples,

\vspace{0.1cm}

\textbf{Obligatory Consumption:} Obligatory consumption encompasses essential expenditures driven by social, cultural, legal, or economic pressures, such as taxes, insurance premiums, tuition fees, loan repayments, and child support. In the digital age, this phenomenon has intensified due to subscription-based models and dependencies on digital ecosystems. A few examples among many include recurring payments that are often overlooked but accumulate to substantial amounts, such as subscriptions for software like Microsoft Office, streaming services like Netflix \cite{EMB2024}, gaming memberships \cite{bercu2016}, and IoT devices relying on subscription-based models \cite{Kerschbaumer2022,Tatsuya22}. Additionally, hidden fees from delivery apps, in-app purchases, and transaction charges further raise consumer costs \cite{Robert24Platform}. Remote work expenses \cite{kenney2016rise} and educational e-resources, such as digital textbooks and e-learning subscriptions, also place additional strain on budgets \cite{Subramaniam2024,Sun2022Digital}.

\vspace{0.1cm}

\textbf{Algorithmic Persuasion \& Impulsive Consumption:} Algorithms embedded in digital platforms employ psychological tactics to influence consumer behavior, often leading to impulsive purchases and financial strain. Consumers have explicitly stated that they attribute excessive consumption to such marketing strategies \cite{pereira2012blame}. For example, strategies such as Buy Now, Pay Later systems encourage overspending by deferring payments \cite{Rosa737080, Thaler1999Mental}, time-limited promotions and personalized recommendations create a sense of urgency and artificial need \cite{Chen2025,Li2022,OberoiG3,Cybertek2024}, dynamic pricing and freemium models restrict free features and encourage in-app purchases \cite{Runge2022,Mishra2017}, social media algorithms amplify the social pressure to conform \cite{Rosa737080, GweeG1, Taylor2022The}, and techniques like fear of missing out (FOMO)-driven scarcity alerts further prompt impulsive buying \cite{Cybertek2024, Zhenga2023}.

\vspace{0.1cm}

\textbf{Work Schedule Instability:} Unpredictable work schedules significantly harm workers' financial security and well-being, particularly affecting vulnerable groups like part-time employees, those with lower incomes or education, and food service and retail sectors with more erratic schedules \cite{schneider2019s, mccrate2018unstable}. The instability of work hours leads to burnout, conflicts between work and personal life, and financial disruptions \cite{hannagan2015income, morduch2017and, farrell2016paychecks, reserve2016report, schneider2017income}. This inconsistency in earnings complicates financial planning, undermining workers' ability to make informed decisions about saving and spending, thus reducing their consumption autonomy. 

At the same time, the increasing use of digital tools and algorithms in workforce management, particularly for labor scheduling, can further undermine workers' consumption agency. While these technologies improve efficiency, they can also lead to greater instability in workers' schedules \cite{kantor2014working,cons_res,Kathleen2019,Zhang2022}. For example, reports from The New York Times highlight that these tools have made work hours more unpredictable, as they can dynamically and in real-time adjust schedules with little to no advance notice \cite{kantor2014working,loggins2020here}.

\vspace{0.1cm}

\textbf{Our Work.} In response to the earlier question of whether individuals today are truly making independent choices or if agency in consumption is becoming more illusionary, we argue that in the age of algorithms, many consumers are subject to \emph{value capture}, where externally imposed values (such as those embedded in manipulative marketing) are adopted without critical reflection or personal adaptation \cite{nguyen2024value}. Furthermore, systemic barriers, including socio-economic inequality, algorithmic biases, and the misuse of algorithms in areas like automated work scheduling, undermine consumer agency by restricting individuals' ability to make decisions aligned with their unique values and needs.

In this context, consumption agency must be recognized as a \emph{value} to be actively cultivated, rather than assumed. Thus, to strengthen agency, it is critical to showcase the impacts of these challenges on consumption autonomy more than ever before. For consumers, understanding these dynamics helps them make more informed decisions about their immediate and future financial needs. They can move toward \emph{true consumption agency}, which enables individuals to make independent decisions about what, when, and how to consume, based on their own values and preferences, without being influenced by external pressures, manipulative strategies, or systemic limitations.

For policymakers, these insights provide valuable guidance for anticipating economic behaviors and developing interventions, such as regulatory frameworks and social welfare programs \cite{frederick2002time}.

Acquiring these insights therefore requires examining consumption behavior (i.e., how individuals make decisions about spending and saving) within the constraints of limited agency. To carry out this study, in this paper, we utilize discounted utility models, a widely recognized class of frameworks in economics for analyzing intertemporal consumption behavior \cite{deaton1992understanding, deaton1989saving}. These consumption models help deepen our understanding of how individuals make decisions about consumption and saving over time.  Our objectives in this study are twofold: \begin{enumerate}
     \item We formally analyze the impact of algorithmic strategies on consumption behavior and the resulting effects on consumer financial well-being and stability. That is, we formally assess consumption with bounded agency where the absence of true agency in consumption decisions leads to financial distress.
     \item We identify solutions that can help mitigate these impacts to foster greater consumption agency.
 \end{enumerate}

We specifically analyze several scenarios to determine whether individuals compelled to adopt predetermined consumption values or influenced by external pressures are at a higher risk of experiencing early financial bankruptcy. Additionally, we examine the adverse effects on financial utility for individuals who, due to uncertainty about their work schedules, are unable to exercise consumption agency or plan their finances effectively. As part of the solutions, we demonstrate how reassessing consumption values (through value deliberation and prioritization of genuine needs) can improve financial well-being and empower consumers.

\vspace{0.2cm}

\textbf{Our contributions.} In summary, the key contributions of this paper are as follows:

\begin{itemize}
    \item We formally demonstrate the consequences of a lack of agency in consumption, which are exacerbated by algorithms. We illustrate these effects across three distinct use cases: obligatory consumption, algorithm-driven impulse purchases, and dynamic work scheduling.

    \item We discuss the significance of true agency in consumption and formally prove that empowering consumers to move toward true agency can help them avoid financial distress.

    \item We highlight the significance of viewing consumption as a value rather than a given in the modern economy. Additionally, we demonstrate the benefits of encouraging deliberation on consumption, which can be facilitated through educational initiatives and policy-driven regulatory measures.

\end{itemize}

We first present the related work in \S\ref{sec:related-work}. We then introduce the core modeling framework in \S\ref{sec:main-model}. Subsequently, we apply this framework to explore issues related to limited agency in several scenarios, including obligatory consumption, algorithmic persuasion and impulsive consumption, and work schedule instability, in \S\ref{sec:obligatory}, \S\ref{sec:impulse}, and \S\ref{sec:work}, respectively. Afterward, we propose potential solutions in \S\ref{sec:solutions} and include a discussion of limitations and future work in \S\ref{sec:limitations}. Finally, we present our conclusions in \S\ref{sec:conclusions}.

\section{Literature Review}
\label{sec:related-work}
\textbf{Consumption Models.}
Consumption models, part of the broader discounted utility (DU) framework, explore how individuals make intertemporal decisions about consumption and saving \cite{deaton1992understanding, deaton1989saving}. These models assume individuals maximize discounted utility, often prioritizing immediate rewards over future ones \cite{chabris2010intertemporal}, and focus on choices regarding when to consume or save \cite{samuelson1937note}. Key models include the permanent income hypothesis (PIH), the life-cycle model \cite{friedman1957permanent, deaton1992understanding, friedman2018theory}, and the neoclassical model \cite{butler2001neoclassical}. The PIH predicts consumption based on expected average income over time, while the life-cycle model adds a finite time frame for asset accumulation and use. The neoclassical model builds on these principles using neoclassical economic theory. Additionally, the income fluctuation problem (IFP), another model within this family of models, introduces income uncertainty and limits consumption to current assets, prohibiting consumption beyond owned resources \cite{ma2020income, sargent2014quantitative, deaton1989saving, den2010comparison, kuhn2013recursive, rabault2002borrowing, reiter2009solving, schechtman1977some}.

\vspace{0.1cm}

\textbf{AI, Consumption, and Agency.} The concept of agency in consumption and in the presence of algorithmic decisions has been explored in various contexts, e.g., studies on consumer agency and its adaptability during environmentally imposed constraints such as the pandemic \cite{gordon2024consumer}, recommendation systems that could either enhance or undermine user agency by tailoring experiences to user preferences but also restricting agency through over-reliance on opaque algorithms \cite{wu2024negotiating, Hosanagar2013}, and studies on fairness in AI emphasizing the importance of agency (as a general concept) as a dimension of equity, noting that constrained or manipulated choices disproportionately affect vulnerable populations \cite{rubel2020algorithms, Alvarez-Garcia2020, nguyen2024value}. Furthermore, prior work suggests that interpretability and transparency can support user agency by fostering trust and enabling informed decision-making \cite{cheong2024transparency, LogRocket2018, Brey2020}.

\vspace{0.1cm}

\textbf{Obligatory Consumption.} Researchers have explored various dimensions of this phenomenon, highlighting how subscription-based models dominate consumer spending. For instance, Amazon Prime and Netflix exemplify this trend, requiring continuous payments for access and locking users into ongoing costs \cite{EMB2024}. Similarly, gaming platforms leverage battle passes and memberships to restrict content access, effectively monetizing user engagement over time \cite{bercu2016}. Additionally, digital ecosystems where consumers are often tied to proprietary infrastructures (such as Apple’s App Store or Google Play) restrict users' ability to switch providers without incurring additional expenses \cite{Shao2020, Geradin2021}. IoT devices, including smart thermostats and home security systems, further exacerbate these dependencies by requiring subscriptions to unlock full functionality \cite{Kerschbaumer2022,Tatsuya22}. Hidden costs associated with digital platforms, e.g., delivery apps that impose service fees and mandatory tips, digital payment systems like PayPal and Venmo with transaction charges are other examples of obligatory consumption \cite{Robert24Platform}. 

Work-related obligations further embed consumption into daily life. The rise of remote work has necessitated investments in digital tools, high-speed internet, and office equipment, while gig economy workers face additional costs for maintaining vehicles or upgrading smartphones to meet platform requirements \cite{kenney2016rise}. Education systems increasingly mandate digital resources, such as e-textbooks and online learning platforms, creating additional financial burdens for students \cite{Subramaniam2024,Sun2022Digital}. Healthcare and wellness sectors have also adopted subscription-based models. Telemedicine services and wearable health devices, such as fitness trackers, often tie essential features to recurring payments, further embedding obligatory consumption into consumers’ lives \cite{Haleem2021,Anawade2024}.

\vspace{0.1cm}

\textbf{Algorithmic Persuasion.} Extensive research has been conducted on persuasive consumption and impulse purchases, particularly within consumer behavior, marketing strategies, and algorithmic design. Algorithmic personalized recommendations, which use tailored products and offers, are commonly employed to boost consumption \cite{Chen2025, Li2022}. These recommendations often make use of user data to enhance relevance and urgency, increasing the chances of unplanned purchases \cite{Chen2025, Li2022}. Time-sensitive offers and scarcity tactics are highly effective in triggering impulsive behavior \cite{OberoiG3, Cybertek2024}. Approaches like flash sales, countdown timers, and low-stock alerts create a sense of urgency and fear of missing out (FOMO), prompting users to make quick purchase decisions without much deliberation \cite{OberoiG3, Cybertek2024}. Dark patterns, which are manipulative design techniques, further contribute to impulse buying. These include deceptive interface designs, preselected add-ons, and hidden costs that push users toward unintended purchases \cite{FTC2022, FTC2023}. Additionally, gamification elements like reward points, badges, and streaks are used to enhance engagement and increase spending, particularly in e-commerce and gaming platforms \cite{Octalysis2024, Scavify2024}.

Social media platforms amplify these effects through influencer endorsements, viral trends, and targeted advertisements, leading to impulsive and unnecessary consumption \cite{Rosa737080, GweeG1}. The integration of ``buy now" features directly within social media apps makes the transition from browsing to purchasing seamless, increasing impulsive buying behavior \cite{APA2019, Singh2024}. The rise of mobile commerce (m-commerce) has also been identified as a significant driver of impulsive buying \cite{WuYe2013, Zhenga2023}. The ease and convenience of smartphones facilitate spontaneous purchases, while push notifications and personalized ads on mobile devices serve as constant prompts for consumption \cite{Zhenga2023, WangWu2023}. Lastly, subscription-based models and microtransactions in digital services and games promote repeated impulse spending through low-cost, recurring payments that accumulate over time \cite{Ferriera2023, Chaudhary2023}.

\vspace{0.2cm}

\textbf{Remark.} Although both obligatory consumption and algorithmic persuasion can involve recurring payments, they differ in nature. Obligatory consumption has many forms, which also include structural, essential, or externally imposed expenses like child support, remote work equipment, telemedicine, online/digital education, and taxes, which inherently differ from subscriptions. These are driven by legal, social, or infrastructural pressures. Detailed examples were listed earlier for each category.

Algorithmic persuasion refers to tactics that exploit consumers to prompt voluntary, often impulsive, spending. While some recurring payments (like in-app purchases) may appear in both contexts, they are not the same. An opt-in subscription under algorithmic persuasion highlights how this expenditure may be strategically engineered for choice-influencing rather than fulfilling essential needs. In contrast, an obligatory subscription is one that is required for basic access to platform lock-ins and needed functionalities (like the ones mentioned earlier that are tied to proprietary infrastructures of App stores or IoT devices /security systems' mandatory subscription fees in digital ecosystems). So, although the subscription categories may appear similar in name, the nature of the expenses differs significantly. Rather than being an impulsive choice, these examples are strict expenses and limit financial flexibility due to economic, legal, or infrastructural constraints.

\vspace{0.1cm}

\textbf{Ruin Analysis and Minimum Subsistence.} Financial bankruptcy, or ruin, has been extensively studied in various theoretical and applied frameworks. Research has focused on evaluating an insurer's insolvency risks \cite{schmidli2002minimizing, asmussen2010ruin}, reducing bankruptcy risks \cite{abebe2020subsidy, papachristou2022allocating}, and investment modeling \cite{Karatzas1986}. These works collectively contribute to understanding and addressing the challenges associated with financial ruin \cite{Bayratkar2012, grandits2015optimal}. Separately, the concept of lower bounds on consumption (minimum subsistence levels) has also been extensively studied in relation to utility-maximizing consumption strategies \cite{zhang2020fairness, ZIMMERMAN2003233, ALVAREZPELAEZ2005633, Shin2011, shim2014portfolio, ANTONY2019124, dwork2018fairness, miranda2023saving, miranda2020model} which indicates that one needs to pay for basic needs such as food and shelter.

\vspace{0.1cm}

\textbf{Work Schedule Instability.} Current research focuses heavily on sociology, particularly on irregular work scheduling and its wide-ranging consequences. Unstable schedules contribute to income volatility \cite{hannagan2015income, morduch2017and, farrell2016paychecks, reserve2016report, schneider2017income}, which in turn leads to financial and life hardships \cite{bania2006income, reserve2016report, leete2010effect, mccarthy2018poverty, Lambert218}. These include burnout from precarious work schedules \cite{schneider2019consequences, Hawkinson2023-jg} and work-family conflicts \cite{golden2015irregular, julia2014}, particularly affecting parents dealing with unpredictable or just-in-time schedules. Furthermore, parental work instability has been linked to increased anxiety and child behavioral problems \cite{schneider2019s}. In the field of Human-Computer Interaction (HCI), researchers have taken a participatory approach to studying similar issues. This includes emphasizing the importance of worker participation in scheduling decisions to ensure fairness \cite{Uhde2020} and exploring elicitation methods to model worker preferences for more effective schedule management \cite{Lee2021}. Additionally, scheduling software and planning algorithms have been identified as contributors to more unpredictable schedules, particularly for low-wage workers in the service sector \cite{kantor2014working, cons_res, Kathleen2019, Zhang2022}. For instance, a New York Times report highlighted cases where algorithmically scheduled employees received their timetables less than three days before the start of the workweek \cite{kantor2014working}. Sudden schedule changes and sales-driven pay reductions have also been correlated with these practices \cite{loggins2020here}.

\section{Fallacy of Agency:  Case Studies}

In this section, we explore various scenarios of consumption behavior that reveal the illusion of agency. We will formally demonstrate that even rational agents can experience adverse outcomes, despite appearing to have clear autonomy. Our analysis focuses on three key cases: 

\begin{enumerate}
    \item  Mandatory consumption of a fixed amount at each time step,
    \item Algorithmic persuasion and impulsive consumption, involving spending money on basic living needs such as food and shelter (minimum subsistence) alongside externally-influenced impulsive consumption values, and
    \item The adverse effects of unpredictable work schedules on worker consumption agency.

\end{enumerate} 

\vspace{0.1cm}

\textbf{Why these three scenarios?} Each condition contributes to a subtle but critical form of epistemic harm in which individuals’ ability to act as credible, autonomous agents is diminished. We selected these three conditions (i.e., obligatory consumption, algorithmic persuasion, and work schedule instability) because they reflect distinct, yet deeply interconnected, ways in which consumer agency is constrained in the digital economy:

\begin{itemize}
    \item Obligatory consumption captures structural and economic constraints that reduce consumers’ freedom to opt out of ongoing financial commitments.
    \item Algorithmic persuasion highlights how digital interfaces and design tactics influence behavior through manipulation.
    \item   Work scheduling instability reflects the temporal uncertainty that undermines individuals’ ability to plan and make informed decisions. 

\end{itemize}

Together, these cases span structural, technological, and temporal dimensions of agency, offering a multidimensional view of how autonomy is eroded. This framework also aligns with recent calls in both economics and ethics in AI toward relational and systemic determinants.

Before delving into these cases, we will first establish the underlying model and system properties.

\subsection{Basic Model and System Properties}\label{sec:main-model}

We begin by outlining the foundational model dynamics underlying our arguments. Suppose a consumer starts with an initial asset value of $a_0$. The consumer aims to consume these available assets while simultaneously working to grow them through saving, investing, and other means. Let $a_t$ be the assets' value at the time-point $t$ and $c_t$ be the consumption. Let $R_t$ be the return on assets and $y_t$ the income where $R_t$ and/or $y_t$ are coming from a known distribution. 

With these parameters established, we can now outline our model of consumption under uncertainty, which serves as the basis for our arguments. Although numerous models aim to capture consumption under uncertainty, they all fundamentally depend on the discounted utility framework, which is widely used in economics to analyze the relationship between consumption and savings (i.e., intertemporal consumption models with discounted utility) \cite{ma2020income, sargent2014quantitative, nokhiz2021precarity, nokhiz2024agent,deaton1992understanding}. In a discounted utility model, an agent consumes an amount $c_t$ at each time step $t$, deriving utility $u(c_t)$ from a concave function $u$. The goal is to determine a policy that maximizes the total discounted utility over time. Formally, we can define the basic model as,
\begin{align}\label{eq:model}
\max \E_{R_t,y_t}\left(\sum_{t=0}^\infty \beta^t u\left(c_t\right)\right) \\
\text{s. t.}\nonumber\\
a_{t+1} = R_t(a_t-c_t)+y_t\label{eq:system-dyn}\\
0 < c_t \le a_t
\end{align}

Here, $\beta \in (0,1)$ represents the discount factor, reflecting consumers' preference for immediate rewards over long-term rewards of equivalent magnitude. The utility function can be any standard utility function in economics, such as the Constant Relative Risk Aversion (CRRA) utility \cite{ljungqvist2018recursive, wakker2008explaining}, which is used to capture individuals' tendency to prefer lower-risk, smaller gains over higher-risk outcomes with potentially larger payoffs \cite{o2018modeling}. The last constraint ensures the consumer can only consume from available assets. This prohibits behaviors such as maintaining positive assets at every time step by borrowing indefinitely.

\vspace{0.1cm}

Using the discounted utility model, we can analyze consumer behavior to understand how individuals act in various circumstances and evaluate their financial stability across different scenarios. To achieve this, we require a concept that encapsulates a consumer's financial instability. In this study, we use the concept of ``ruin" to capture this notion. Ruin refers to a situation where a consumer exhausts their assets within a short period. Formally:

\vspace{0.1cm}

\textbf{Ruin.} We define ruin as the state where $\exists T < \infty$ such that $a_{t+1}\le 0$ for $t=T$. This would mean that the consumer would no longer be able to sustain their consumption, leading to financial ruin.

\vspace{0.2cm}

In \S\ref{sec:obligatory} and \S\ref{sec:impulse}, we examine two scenarios: obligatory consumption and impulsive consumption where various behavioral constraints impact consumers' financial stability. We demonstrate how the probability of financial ruin shifts under different conditions and how these behaviors can lead to early ruin. In \S\ref{sec:work}, we extend the analysis to temporal factors, illustrating how greater certainty in temporal events can enhance achieved utility.

\subsection{Obligatory Consumption}
\label{sec:obligatory}
In this section, we examine the first use case: obligatory consumption. 

\vspace{0.2cm}

\textbf{How does this condition affect consumer agency?} Obligatory consumption constrains agency by turning certain financial decisions into non-choices. When individuals are locked into certain payments for mandatory needs, utilities, or platform-related fees (often without a clear understanding of long-term costs) they lose the ability to re-prioritize their expenditures in response to changing needs or circumstances. These expenses become effectively fixed, and the inflexibility limits the individual’s capacity to act in accordance with their own values or preferences. Even when expenses were initially chosen, their persistent and sometimes opaque nature (e.g., hidden mandatory fees) undermines ongoing deliberation and adaptive decision-making, weakening the authenticity of consumer autonomy (which is an epistemic harm).

To model this scenario, the consumer is restricted to consuming a predetermined fixed amount at each time step. While the consumer has the flexibility to select this fixed amount initially, it cannot be adjusted later, effectively committing to a rigid consumption plan. We will formally investigate the effect this would have on the financial stability of the consumer.

To make things simple, let us assume $R_t=1$ for all $t$ and $y_t$ is coming from a distribution $D_Y$. We can show the following,

\vspace{0.1cm}

\begin{theorem}\label{thm:oblig-ruin}
Consider an income $y_t$ coming from a distribution $D_Y$ with a known mean $Y$. Assume the consumer is only allowed to consume a fixed constant obligatory amount $c \ge 0$ at every step. A rational utility-maximizing consumer would go to ruin when the utility function is concave.
\end{theorem}

\vspace{0.1cm}

 \textbf{Concept Sketch.} Before presenting the formal proof of Theorem~{\ref{thm:oblig-ruin}}, we provide a brief overview of the main argument. In the limit, fixed consumption converges to the income level. Given the concavity of the utility function, Jensen's inequality implies that the resulting utility over an infinite horizon is bounded. This observation allows us to show that, for sufficiently large initial assets, there exists a finite-horizon consumption path that yields a strictly higher utility, ultimately leading to termination in a finite number of steps.

\vspace{0.1cm}

\begin{proof}
Consider the discounted utility model introduced in Equations~\ref{eq:model} and~\ref{eq:system-dyn}, under the setting where $R_t=1$ for all $t$. Assume the consumer goes to ruin at time $T$. Let $c_T$ be the constant consumption for ruin at time $T$. Note that,
\begin{align*}
c_T = \sum_{t=0}^{T-1} \frac{y_t}{T}+\frac{a_0}{T}
\end{align*} 
Let $c= c_\infty$. Note that $E(c) \rightarrow Y$.

Consider the scenario where $T \rightarrow \infty$. We can see that,
\begin{align*}
E \left(\sum_{t=0}^\infty \beta^t u(c_t)\right) &= \sum_{t=0}^\infty \beta^t E\left(u(c)\right)\\ &\le \sum_{t=0}^\infty \beta^t u\left(E\left(c\right)\right)
\end{align*}
where the inequality comes from Jensen's Inequality. Therefore, we get that $E \left(\sum_{t=0}^\infty \beta^t u(c_t)\right) \le \frac{1}{1-\beta}u(Y)$. Note that  $c_t=c$ for all $t$ due to an obligatory amount $c$. 

Now consider the scenario where $a_0 > u^{-1} \left(\frac{1}{1-\beta}u(Y)\right)$. However, in this case, consuming $a_0$ at the first time step would leave the consumer with more utility than not going to ruin. A rational consumer with the goal of maximizing utility can, therefore, consume all assets in a short finite time and go to ruin but still end up with better overall utility compared to a scenario where they would save assets and consume within an infinite horizon.

\end{proof}

This shows that no matter how rational and optimal the consumer's behavior is, the lack of agency and the obligatory nature of the consumption could lead to adverse outcomes. Note that we have only relied on the fact that the utility is concave. This is a characteristic common to the class of consumption functions referred to as \emph{risk-averse}. This indicates that even risk-averse behavior is insufficient to prevent a consumer from facing ruin when their agency is constrained. 

\vspace{0.1cm}

\textbf{Remark on the Concave Utility.} The use of concave utility functions in settings with fixed or non-discretionary consumption is a preexisting choice. For example, \cite{angoshtari2022optimalconsumptionhabitformationconstraint} study a habit-formation constraint using CRRA utility, where consumption cannot fall below a fraction of past levels. The work of \cite{Xu_2016} consider a fixed upper bound on consumption under CRRA preferences, while the work of \cite{Roh2020} analyze both fixed upper and fixed lower consumption bounds using the same class of utility functions. These examples illustrate that concave utility functions have been used in models where consumption is partially or fully constrained, similar in spirit to our notion of obligatory consumption.

\vspace{0.1cm}

However, in practice, a consumer cannot consume arbitrarily small amounts near $0$. To meet their basic needs and sustain themselves (i.e., minimum subsistence), the consumers must consume at least a specific minimum amount at each time step, which serves as a lower bound on consumption. Our results trivially extend to the case where minimum subsistence is considered. We will state this as a corollary for completeness.

\vspace{0.1cm}

\begin{corollary}
Consider an income coming from a distribution $D_Y$ with a known mean $Y$. Assume the consumer is only allowed to consume a fixed constant amount $\ge b$ at every step (where $b$ is the minimum subsistence). A rational utility-maximizing consumer would go to ruin when the utility function is concave.
\end{corollary}

\begin{figure*}
    \centering
\includegraphics[width=0.8\columnwidth]{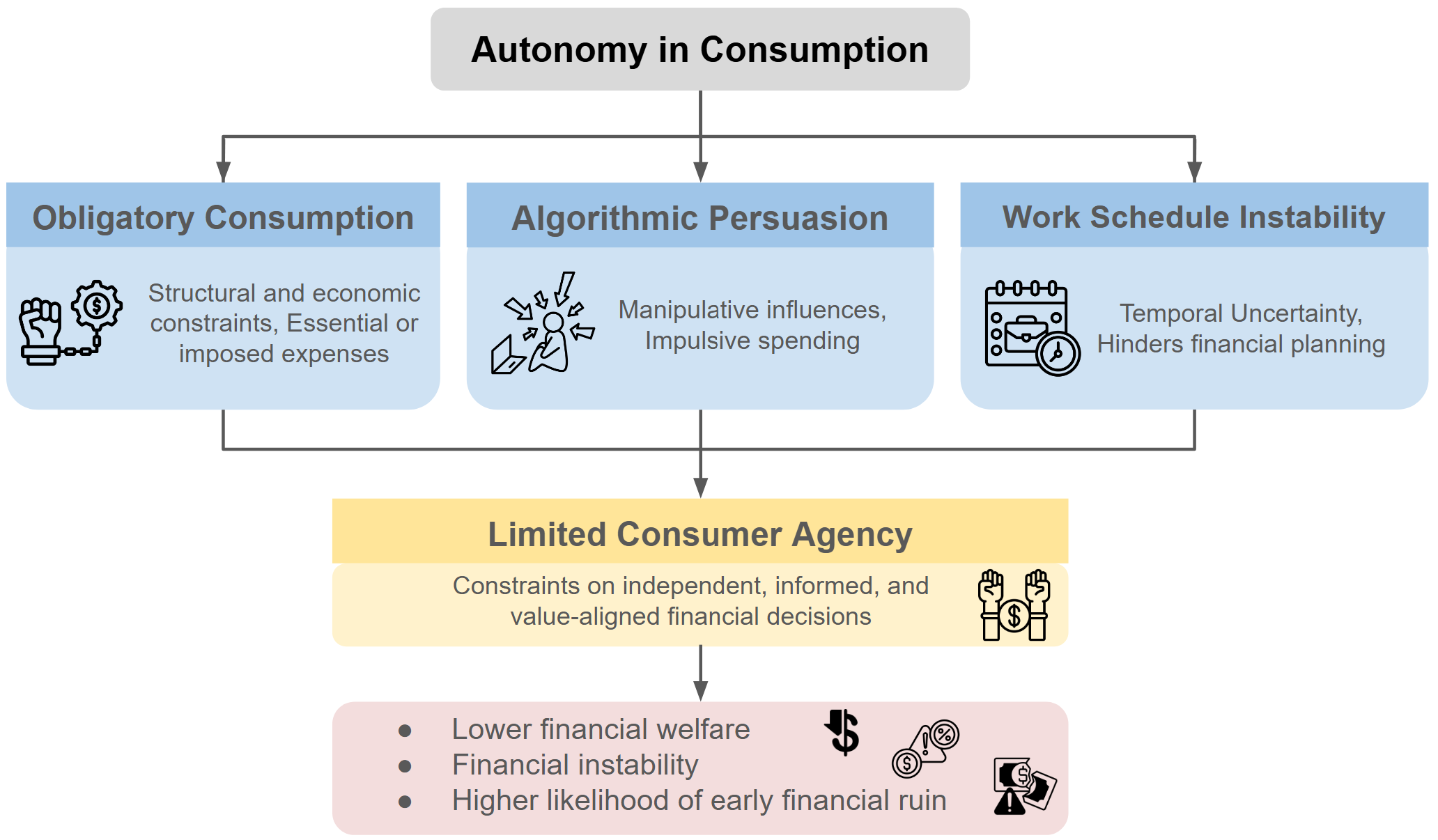}
    \caption{A summary of the explored real-world scenarios with bounded consumer agency, with each scenario illustrating a distinct aspect of limited consumer agency within the digital economy, often leading to unfavorable financial outcomes.}
    \label{fig:conceptual}
\end{figure*}
\subsection{Algorithmic Persuasion and Impulsive Consumption}
\label{sec:impulse}
In this section, we examine the case of algorithmic persuasion and impulsive consumption. 

\vspace{0.2cm}

\textbf{How does this condition affect consumer agency?} Through personalized advertising, urgency cues, and psychologically manipulative interface design, digital platforms create an environment where impulsive spending feels like self-directed behavior but is heavily externally shaped. Consumers would be responding to engineered triggers designed to bypass critical deliberation. This epistemic harm produces what \cite{nguyen2024value} calls value capture, where external systems impose consumption values that individuals adopt without recognizing the deviation from their own long-term goals or ideals. Thus, consumer agency is substantially hollowed out.

In this setting, the consumer is allowed to consume any (feasible) amount. The consumer consumes $c_t$ at each time point $t$ which consists of the minimum subsistence $b_t$ and the impulsive consumption $c_t-b_t$. The consumer's goal is to maximize their utility. We can see that under certain constraints, even when the consumer is assumed to have freedom of choice, the optimal behavior for the consumer is to consume available assets in a short finite time to gain the maximum utility possible.

\vspace{0.1cm}

\begin{theorem}\label{thm:con-ruin-min}
Consider the scenario where the consumer has control over their consumption, subjected to the consumption $c_t \ge b_t$ at each time point $t$ (where $b_t$ comes from a distribution with mean $B$ and the income comes from a distribution with mean $Y$). Assume the return rate $R_t$ is fixed at 1, and for all $t$, $|b_t-B| \le \delta$ and $|y_t-Y| \le \epsilon$. Also, assume $(Y-B) < 0$. There exists a $T < \infty$ such that with high probability $a_T \le 0$.
\end{theorem}

\vspace{0.1cm}

\textbf{Concept Sketch.} Before presenting the formal proof of Theorem~{\ref{thm:con-ruin-min}}, we begin with a brief outline of the main idea. Since both income and minimum subsistence levels are bounded, we can derive bounds on the change in assets at any time step $t$. Because the final asset level is the cumulative sum of these changes over time, we can apply this boundedness along with standard concentration inequalities. Under the assumption that the means of income and minimum subsistence behave as expected, it follows that, with high probability, the agent will reach ruin within a finite number of steps $T$. The formal proof is given below.

\vspace{0.1cm}

\begin{proof}
Consider the model introduced in Equations~\ref{eq:system-dyn} and~\ref{eq:model}. Let $a_0$ be the initial assets and $a_T$ be the assets at time point $T$. Then we can see that,
\begin{align*}
a_T &= a_0 +\sum_{t=0}^{T-1} y_t-\sum_{t=0}^{T-1} c_t \\
 &\le a_0 +\sum_{t=0}^{T-1} y_t-\sum_{t=0}^{T-1} b_t
\end{align*}
where the last inequality comes from the fact that $c_t \ge b_t$ for all $t$. For ease of notation, let $\lambda_T = a_0 +\sum_{t=0}^{T-1} y_t-\sum_{t=0}^{T-1} b_t$. Note that $\lambda_T \ge a_T$. Also,
\begin{align*}
E(\lambda_T)
 \le a_0 +(Y-B)T
\end{align*}
We can also see that under our assumptions, 
\begin{align*}
Y-B-\delta-\epsilon \le y_t-b_t \le Y-B+\delta+\epsilon
\end{align*}

\vspace{0.1cm}

Using the Hoeffding's inequality on $\lambda_T$ and assuming $T$ is large enough, i.e., $T$ is such that $(B-Y)T \ge a_0$, we see that,
\begin{align*}
& Pr \left(a_T \ge 0\right) \\&\le Pr \left(\lambda_T \ge 0\right) \\
& = Pr \left(\lambda_T-\left(a_0+(Y-B)T\right) \ge -a_0+(B-Y)T\right) \\
& \le e^{-\frac{2\left(-a_0+(B-Y)T\right)^2}{\left(Y-b+\delta+\epsilon-\left(Y-b-\delta-\epsilon\right)\right)^2 \cdot T}}\\
& = e^{-\frac{2\left(-a_0+(B-Y)T\right)^2}{\left(2\delta+2\epsilon\right)^2 \cdot T}}
\end{align*}
Furthermore, for any $T$ such that $T > 2a_0/(B-Y)$, we get,
\begin{align*}
Pr \left(a_T \ge 0\right) &\le e^{-\frac{2\left((B-Y)T/2\right)^2}{(2\delta+2\epsilon)^2 \cdot T}} \\
& = e^{-\frac{\left((B-Y)\right)^2 T}{8(\delta+\epsilon)^2}}
\end{align*}

Let $c = \frac{\left((B-Y)\right)^2}{8(\delta+\epsilon)^2}$. We get that with probability $\ge 1-e^{-cT}$, the assets at time $T$, $a_T$, satisfy $a_T \le 0$ and the consumer would go to ruin at time point $T$.

\end{proof}

We observe that the introduction of impulse consumption and minimum subsistence constraints in both of our settings can eventually lead to ruin under specific conditions, as outlined in the previous claims. Theorem~{\ref{thm:con-ruin-min}} shows that in this case, agents are likely to experience ruin within a finite number of steps. Moreover, the probability of avoiding ruin decreases exponentially with time $t$.

It is important to note that the implications of Theorem~\ref{thm:con-ruin-min} are significant: \begin{itemize}
    \item This is because when income distributions have low average income and high variance, even if income is high at certain points in time, the likelihood of early ruin is high as well. In other words, despite having high income at certain moments, the high uncertainty in income means that if regular consumption exceeds the average income, there is a strong likelihood of quickly depleting resources.
    \item Even if the average income is lower than the consumer's expenditures, they may still recklessly deplete their savings to cover excessive consumption. As discussed earlier, algorithms can push individuals to consume, even when financially incapable, through targeted ads, creating urgency or scarcity (such as FOMO-driven strategies), freemiums, and time-limited offers. These tactics exploit psychological and emotional triggers to influence decision-making, a phenomenon known as optimism bias, where individuals tend to underestimate their expenses, often leading them to make poor financial decisions despite having the means to plan more effectively \cite{peetz2009budgetingbias}.

\end{itemize} 

\subsection{Work Scheduling}
\label{sec:work}

In the previous sections, we observed that an individual's ability to maintain long-term financial stability varies based on the level of agency they possess. However, in our results so far, adjustments to the agency were made by modifying monetary parameters. In this section, we explore the concept of agency as a temporal parameter. This idea has been studied in recent work~\cite{nokhiz2024modeling, nokhiz2025counting, nokhiz2025counting1}, which introduce the concept of ``lookahead", representing the degree of certainty one has about future outcomes. This is especially relevant as we would like to formally evaluate the hypothesis that having more awareness of upcoming work schedules enables workers to better manage their finances and exercise greater autonomy. That is:

\vspace{0.2cm}

\textbf{How does this condition affect consumer agency?} Unstable work schedules impair consumption agency by undermining the individual’s ability to plan, anticipate, and allocate financial resources effectively. Agency is not only about making choices, but about having the temporal and informational stability to evaluate options. When income is volatile and work hours unpredictable (as often happens under algorithmic labor scheduling), individuals lack the foresight necessary for financial planning. This temporal uncertainty prevents them from aligning spending and saving with their values, instead forcing short-term, reactive decisions. The formal analysis in this section shows that even rational consumers with unstable schedules face systematically lower utility, illustrating that agency is structurally constrained by temporal unpredictability.

A lookahead of $k$ steps (weeks) implies that the agent is aware of the exact income and financial shocks they will face over the next $k$ steps in their work schedule. Research in \cite{nokhiz2024modeling, nokhiz2025counting,nokhiz2025counting1} formally demonstrates the consequences of this concept, as explained below:

\vspace{0.1cm}

\begin{theorem}[\cite{nokhiz2024modeling,nokhiz2025counting,nokhiz2025counting1}]\label{thm:main-lookahead}
Consider two individuals, one with a lookahead of $k$ steps and one with no lookahead. There are instances with bounded income, where the individual with lookahead achieves significantly higher utility compared to the individual without lookahead. Furthermore, the gap between the utility of the individual with lookahead and the individual without lookahead linearly increases with $k$. Formally, if $u_c$ is the utility of the individual with lookahead $k$ and $u_z$ is the utility of the individual without lookahead, then 
\[u_c-u_z \ge \Omega(k)\]
\end{theorem}

\vspace{0.1cm}

The formal version of the Theorem~\ref{thm:main-lookahead}, and the complete (restated) proof from~\cite{nokhiz2024modeling, nokhiz2025counting,nokhiz2025counting1}, are available in Appendix B, for completeness and ease of reference (The appendix to this paper is available after the references.). The proof itself uses a carefully defined class of income distributions and applies Yao's min-max theorem~\cite{motwani-raghavan} to focus the argument on deterministic algorithms, to demonstrate the desired utility gap.

\vspace{0.1cm}

Theorem~\ref{thm:main-lookahead} implies that even if the worker without lookahead is behaving optimally, the worker with lookahead consuming optimally is always guaranteed to have an advantage. That is, the gap between the utilities of the worker who knows their schedule in advance with that of the worker without a lookahead, increases with $k$. Therefore, the more lookahead the individual has, the more utility they can acquire. This also suggests that a worker who is aware of their work schedule will consistently experience greater utility from consumption, as they can manage their finances more effectively and exercise greater autonomy.

To visually summarize the discussions in {\S\ref{sec:obligatory}}, {\S\ref{sec:impulse}}, and {\S\ref{sec:work}}, Figure {\ref{fig:conceptual}} shows a conceptual diagram/flowchart of all scenarios, their characteristics, and potential outcomes.

Additionally, some experimental results with real-world data inspirations \cite{fedres, minsub, patnaik2022role, cooper2014discounting, ericson2019intertemporal} on the scenarios discussed in these sections are presented in Appendix D. Note that the experiments are included solely to illustrate key theoretical insights in action; they are not the focus of the contribution. Their role is supportive and interpretive, not foundational.

\section{Suggested Solutions: Intervening by Consumer Value Deliberation}

\label{sec:solutions}

The results presented highlight a nuanced outlook: while the intuitive understanding is that a lack of control is detrimental, the precise formal long-term consequences are often uncertain and rarely discussed. Our results formally demonstrate the significance of these aftereffects: Many consumers fall prey to value capture, where externally imposed values (such as those embedded in manipulative marketing) are accepted without critical reflection or personal adaptation \cite{nguyen2024value}. Additionally, they encounter systemic barriers such as mandatory consumption and suboptimal practices like faultily executed automated work scheduling, which further reduce consumer agency.

Thus, we argue that consumer agency is a \emph{value} to be nurtured rather than an inherent given.  These theoretical insights (particularly those addressing situations intensified by algorithmic influences) serve as a tangible reminder that fostering consumption as a value holds greater significance in today's economy. 

\subsection{True Consumption Agency}

Fostering agency can be seen as moving toward true agency which means making independent, informed, and deliberate choices about consumption, free from external pressures, manipulative tactics, or systemic constraints. True agency allows individuals to make decisions based on their genuine needs and financial circumstances, resisting the forces that encourage unnecessary or excessive consumption.

\vspace{0.1cm}

In an ideal world, where individuals enjoyed this level of agency and were trying to maximize their financial utility, they would be capable of avoiding financial ruin indefinitely (we formally prove this in Theorem 5 in Appendix A). However, true consumption agency may be an overly idealistic concept, given the reality that certain expenses, such as food, rent, and legal obligations like taxes and child support, are unavoidable and essential for maintaining a basic standard of living. While true agency may be an aspirational ideal, striving toward it as a value provides a framework for developing practical solutions that enhance consumer autonomy.

\vspace{0.1cm}

Therefore, we propose \emph{consumption value deliberation as a guiding mindset} for developing interventions at both individual and societal levels. This approach seeks to address unavoidable financial obligations while empowering consumers to exercise greater consumption autonomy. Value deliberation empowers consumers by enabling them to make more informed choices. This process, for both consumers and policy-makers, involves actively evaluating competing values, needs, and preferences that shape spending and saving decisions, thereby taking incremental steps toward true agency.

\subsection{Individual Consumer Value Deliberation}
\label{sec:main-intervention}

At the individual level, value deliberation-based interventions lead to consumer empowerment by enabling individuals to \emph{consciously prioritize their genuine needs above all else}. This type of intervention is inspired by previous work emphasizing that budgeting helps with a clearer understanding of one's financial situation, leading to better decision-making and increased financial stability \cite{archuleta2013budgeting, de2020impact, neubauer2017financial}. Examples include,

\vspace{0.2cm}

\textbf{Encourage Consumer Reflection and Awareness: }Inform consumers on how algorithmic marketing manipulates their choices. Promote critical reflections to help consumers identify and prioritize their authentic needs.

\vspace{0.1cm}

\textbf{Promote Financial Literacy and Resistance Strategies:} Equip consumers with skills that help them with budgeting and resistance to spending on the fly on algorithmic nudges.
Develop educational resources to help consumers regain control over their decision-making processes in the digital environment.

\vspace{0.1cm}

\textbf{Controlled Use of Algorithmic Suggestions:} Advocate for treating algorithmically-generated recommendations as tools for exploration rather than definitive guides. Consumers should use recommendation systems critically, adapting them to their personal values rather than adopting them wholesale \cite{forbes2024ethical}.

\vspace{0.1cm}

\vspace{0.1cm}

To assess this category of interventions, we formally demonstrate their impact. Specifically, we assume that individuals allocate their financial resources initially to non-negotiable essential expenses (such as food, shelter, taxes, transportation) and other fundamental and obligatory needs. Any remaining funds are then budgeted for discretionary (opt-in) purchases, which are frequently shaped by algorithm-driven influences.

Consequently, let $b$ be the sum of minimum subsistence, obligatory consumption, and a fixed set of opted-in impulsive consumptions. As long as the remaining variable fraction of the impulsive consumption is $< a_t+y_{t+1}-b$ the agent can avoid ruin. That is,

\vspace{0.1cm}

\begin{theorem}\label{thm:ture-agency}
For simplicity, assume $R_t = 1$. Also assume $\beta > 1/2$. Let $y_t \ge b$ (where $b$ is the fixed non-negotiable expenses) for all $t$. Then the optimal consumption of a utility-maximizing consumer would allow them to avoid ruin. 
\end{theorem}

\vspace{0.1cm}

\textbf{Concept Sketch.} Before presenting the formal proof of Theorem~{\ref{thm:ture-agency}}, we briefly outline the core idea. Under the given assumptions, we show that if there exists a consumption path that terminates at some finite time $t$, it is possible to construct an alternative infinite-horizon consumption sequence that yields strictly higher utility. The formal proof is given below. 

\vspace{0.1cm}

\begin{proof}
Assume otherwise. The agent would have to stop consuming at some time point $T$ for this to be the case. Now consider the optimal consumption sequence $\{c_0,c_1,c_2,\dots,c_T\}$.

\vspace{0.2cm}

{\bf Case 1}: If $c_T \ge 2b$.

In this case consider the new consumption sequence, $\{c_0,c_1,c_2,\dots,c_T-b,b,b,b,\dots\}$ where the agent consumes $c_T-b$ at time step $T$ and keeps on consuming $b$ after that. It is easy to see that this is a feasible solution. We can see that the change in utility, $\delta_u$, is 
\begin{align*}
\delta_u &=\left( \sum_{t=0}^{T-1} \beta^t u(c_i)+\beta^T u(c_T-b)+\sum_{t=T+1}^\infty \beta^t u(b)\right)\\
&-\left( \sum_{t=0}^{T-1} \beta^t u(c_i)+\beta^T u(c_T)\right)\\
&=\beta^T \left(u(c_T-b)+\frac{\beta}{1-\beta}u(b)-u(c_T)\right) \\
& \ge \beta^T \left(u(c_T-b)+u(b)-u(c_T)\right) \\
& \ge 0
\end{align*}
The first inequality simply uses the fact that $\beta > 1/2$ and the second inequality is due to the fact that $u(.)$ is concave and therefore $u(c_T-b)+u(b) \ge u(c_T)$. Thus, in this case, we get a consumption sequence of infinite length that yields better utility.
\vspace{0.2cm}

{\bf Case 2}: If $c_T < 2b$.

We can see that there exists a $d > 0$ such that $(1-\beta)c_T+\beta d = b$ (because $1-\beta < 1/2$ and $c_T < 2b$). From this, using the concavity of $u(.)$, we can see that $u(b) = u((1-\beta)c_T+\beta d) \ge (1-\beta)u(c_T)+\beta u(d) \ge (1-\beta)u(c_T)$. Now we can consider the new consumption sequence where $\{c_0,c_1,c_2,\dots,c_{T-1},b,b,b,\dots\}$, where the agent keeps on consuming $b$ after time step $T-1$. It is easy to see that this is a feasible solution. The change in utility would be,
\begin{align*}
\delta_u &=\left( \sum_{t=0}^{T-1} \beta^t u(c_i)+\sum_{t=T}^\infty \beta^t u(b)\right)\\
&-\left( \sum_{t=0}^{T-1} \beta^t u(c_i)+\beta^T u(c_T)\right)\\
&=\beta^T \left(\frac{1}{1-\beta}u(b)-u(c_T)\right)\\
& \ge 0
\end{align*}
The inequality comes from the fact we established earlier. Therefore, in this case, we get a consumption sequence of infinite length that yields better utility. We can observe that in either case, for any finite consumption sequence, there is a way to improve it with an infinite horizon. Therefore, with a sufficiently high income to cover their basic needs, consumers can maximize their utility while still avoiding early ruin.

\end{proof}

To summarize, by deploying interventions that increase self-discipline to allocate the funds for the products or services truly needed (and then evaluating the remaining amount for discretionary spending), consumers can gain greater control over their finances and avoid early ruin. 

\vspace{0.2cm}

\textbf{Remark on Return Rates and $\beta$.} 
Theorem~{\ref{thm:ture-agency}} assumes $R_t = 1$ and $\beta > 1/2$. The assumption $R_t = 1$ is made for simplicity, but the argument does not depend on this specific value and remains valid for any $R_t \ge 1$. The condition $\beta > 1/2$ reflects the idea that when $\beta$ is too small, the incentive to delay consumption diminishes, leading consumers to prefer immediate consumption in order to maximize utility. Ensuring that $\beta$ is sufficiently large prevents the resulting behavior from being solely driven by discounting effects.

\subsection{Ethical and Social Regulations}

On a societal level, consumption value deliberation can happen by policy-makers implementing regulations that recognize the importance of respecting consumer autonomy, enhancing their agency, and ultimately guiding them toward achieving true consumption agency. Examples include,

\vspace{0.1cm}

\textbf{Promote Ethical AI Development:}
Encourage the creation and implementation of ethical AI systems that prioritize consumer and worker well-being to avoid disturbing the delicate balance between creating persuasive advertising and exploiting the customer \cite{forbes2024ethical}. This can be achieved by setting clear ethical guidelines \cite{forbes2024ethical} to prevent manipulative tactics and providing incentives (such as enhancing employee morale and gaining a competitive edge) to companies that prioritize responsible AI practices, particularly in areas like work scheduling.

\vspace{0.1cm}

 \textbf{Algorithm Regulation:} Regulating algorithms can help reduce uncertainties related to financial planning. For instance, the Schedules that Work Act of 2014 (H.R. 5159), introduced in the U.S. Congress, requires employers to provide employees with at least two weeks' notice of schedule changes and minimum hours \cite{golden2015irregular}. This Act acknowledges that employees should have the ability to make informed decisions about their personal time and should not be forced to adapt to sudden changes imposed by employers. This value is rooted in respect for individual rights and personal choice. 

\vspace{0.1cm}

In another example, the San Francisco Board of Supervisors has improved protections for hourly workers in retail chain stores by adopting provisions from the Retail Workers Bill of Rights. These new regulations require employers to provide greater advance notice when setting/changing work schedules, aiming to improve predictability for employees \cite{golden2015irregular}.

\vspace{0.1cm}

In the context of work scheduling, implementing regulations like the examples above has shown positive effects. Studies like \cite{nokhiz2024modeling, nokhiz2025counting,nokhiz2025counting1} reveal that employees who are given advanced notice of schedule changes experience improved financial well-being. Advanced planning allows workers with variable schedules to gain greater autonomy over their financial planning. That is, they can manage their consumption and savings more effectively, which leads to better financial decision-making and agency.

\section{Limitations and Future Work}
\label{sec:limitations}

In this section, we present discussions and list some of the limitations of our study.

\vspace{0.1cm}

\textbf{Intertemporal Consumption Models.} The model employed here is intentionally simple, relying on utility-maximizing rational agents. This does not imply that we believe individuals in the real world are perfectly rational utility maximizers. Rather, our focus is on demonstrating that even within this straightforward representation of reality, where agents make the most optimal decisions possible, the adverse effects of a lack of agency in consumption are still evident. Furthermore, while discounted utility models may not be as innovative or advanced as some newer alternatives, they remain the most commonly used framework for analyzing intertemporal choice \cite{cohen2020measuring}.

\vspace{0.1cm}

However, the absence of debt and other liabilities is a limitation in our framework. For example, in real-world scenarios, controlled debt, such as credit card usage, significantly influences consumption patterns. Consumers often rely on credit even when they lack sufficient assets to cover their expenses. However, our model does not account for these scenarios. Debt plays a dual role in consumption: it acts as an enabler, allowing individuals to make purchases they might otherwise forgo, while simultaneously introducing hidden costs through the accumulation of interest and long-term debt burdens. Additional details on debt (and discounting \cite{ainslie1992hyperbolic, frederick2002time}) can be found in Appendix C.

\vspace{0.1cm}

Additionally, consumers may not only be individuals but also organizations or communities, each aiming to maximize multiple objectives that may not necessarily align with utility. For instance, they might prioritize minimizing expenses or risk, maximizing social welfare, environmental sustainability, profit, and so on. Our model does not account for these cases.

\vspace{0.1cm}

Furthermore, when using consumption models to analyze intertemporal decision-making, it is crucial for researchers, policymakers, and advisors to account for the normative commitments that may be overlooked \cite{laufer2023optimization, nokhiz2025rethinking, nokhiz2017understanding}. For example, incorporating ethical considerations, such as the impact of advertising or the psychological effects of social comparison, could lead to more comprehensive models.

\vspace{0.1cm}

Future research in this area could explore how to incorporate these normative dimensions into optimization models. This could involve developing frameworks that not only account for traditional economic factors but also consider the ethical implications of consumption decisions. 

\vspace{0.2cm}

\textbf{Behavioral Economics, Cognitive Impacts, and Prospect Theory.} While our model is grounded in the rational utility-maximizing framework (due to the reasons mentioned in the previous paragraph), we acknowledge that real-world consumer behavior gets impacted by cognitive factors and behavioral elements. Behavioral economics has demonstrated that individuals frequently exhibit present bias \cite{laibson1997golden}, loss aversion and reference dependence \cite{kahneman1979prospect}, and inconsistent preferences under uncertainty \cite{tversky1992advances}. These tendencies can significantly impact consumption decisions. For instance, present-biased consumers may be particularly vulnerable to short-term promotions or Buy Now Pay Later schemes, choosing immediate gratification despite long-term financial harm \cite{o2011time}. Similarly, the framing of losses in subscription-based services or scarcity cues in digital marketing can exploit reference dependence, causing consumers to over-consume or stick to suboptimal payment commitments. Prospect theory \cite{kahneman1979prospect} provides a more psychologically realistic alternative by capturing such anomalies through non-linear probability weighting and asymmetric value functions.

While these extensions would undoubtedly enrich the descriptive accuracy of our framework, our decision to use a classical discounted utility model was deliberate: it allows us to isolate the structural impacts of agency constraints in a clean and tractable way, ensuring that any observed harms are not artifacts of irrational behavior but instead emerge even under optimal decision-making. This offers a tractable estimate of the risks posed by limited agency, meaning that under behavioral models, these effects may be stronger or more pervasive. Future work could build on this foundation by embedding bounded rationality and behavioral parameters into the model to evaluate their compounding effects.

\vspace{0.1cm}

 \textbf{Societal Uniformity.} Our paper assumes a level of societal uniformity with all agents facing similar external forces and constraints, however, there exists societal inequities where certain groups are targeted more by algorithms or have more payment obligations. Furthermore, the likelihood of benefiting from public policies aimed at addressing these disparities may also vary among different demographic groups. For example,  if the individual is a gig worker, some of the policy interventions such as a minimum window of advance notice for altering work schedules do not apply to them. If there exists diversity among individuals in managing their finances based on a specific community they are part of, the system could try to incorporate this heterogeneity into its decision-making processes.

 Additionally, lower-income individuals, gig workers, or those with limited digital literacy, are more likely to rely on platforms with dark patterns and  Buy Now Pay Later systems. These dynamics are further compounded by racial and ethnic disparities, as marginalized communities may be overrepresented in precarious labor sectors and underprotected by regulatory infrastructure \cite{schneider2017income}. Thus, the risks we model may in fact be amplified for these groups, making the normative question of consumption agency all the more urgent.

This recognition has direct implications for the design of targeted interventions: First, policy frameworks could prioritize sector-specific algorithmic regulation, particularly in domains like retail, delivery, and food service, where unpredictable scheduling and surveillance technologies are common. For instance, existing initiatives like the Retail Workers Bill of Rights in San Francisco can be endorsed more carefully for the most affected workers. Second, consumer protection strategies can be sensitive to digital manipulation risks among lower-income populations, like opt-in delays for  Buy Now Pay Later usage on platforms where these features disproportionately exploit financial vulnerability. Third, financial education efforts must be context-aware: programs targeting low-income or gig workers should not assume ample time, attention, or technological access. Instead, these programs can be delivered via workplace partnerships, mobile-accessible content, or simplified budgeting tools that accommodate erratic income. Finally, consumer-facing transparency requirements (for example, algorithmic decision notices or cost-accumulation alerts) could help mitigate manipulation in environments where deliberation is otherwise undermined.

\vspace{0.1cm}

\textbf{Other Future Directions: Macro Socioeconomical Impacts.} There are external socioeconomic impacts on one's income that we have not considered as deciding factors in our study, e.g., inflation erodes the purchasing power of consumers, making it more difficult for individuals to maintain their standard of living and save for future consumption. Similarly, unexpected social emergencies, such as pandemics, can lead to financial precarity with sudden unemployment, reduced income, or increased health-related expenses \cite{nokhiz2021precarity}, significantly affecting consumption patterns.

Like other foundational economic models, ours offers a foundation for the future. More complex extensions (with more complex macro factors) are possible. So to expand upon and complement the aforementioned macro-level factors, one may consider:
\begin{itemize}
    \item Inflation: As mentioned earlier, inflation reduces the real value of income and savings over time, potentially accelerating financial ruin even under otherwise sustainable consumption plans.
\item Economic Recessions: Lead to increased unemployment, reduced wages, and greater scheduling instability, further constraining agency and raising the likelihood of ruin.
\item Public Health or Emergency Crises (e.g., pandemics): Trigger sudden behavioral changes in spending (e.g., panic buying), disrupt work schedules, and create new obligatory expenses (e.g., healthcare and remote work infrastructure).

\item Government-based regulations like subsidies or stimulus packages can also influence income levels and consumption decisions by providing temporary relief or incentives for certain types of spending \cite{nokhiz2024agent}, such as renewable energy adoption or housing investments.

\item Also, economic recessions or global supply chain disruptions may result in income instability and increased costs of goods and services \cite{nber2024}, further constraining consumer choices.
\end{itemize}

 These impacts and the underlying population's response to them (as well as the recourse to change circumstances \cite{gupta2019equalizing}) differ significantly based on factors such as gender \cite{blau2017gender, reader2022models}, disability \cite{whittaker2019disability}, race \cite{zwerling1992race}, and health status \cite{hicken2014racial}. While these factors are critical in shaping consumption behavior, they fall outside the scope of our current study and remain important areas for the future. 

\vspace{0.2cm}

\textbf{On the Motivation Behind Proposed Interventions \& More Specific Interventions.} To ground our recommendations in real-world practices, we draw on existing policy interventions as illustrative cases. For instance, the Schedules that Work Act in the U.S. Congress and San Francisco’s Retail Workers Bill of Rights demonstrate the viability of regulating algorithmic work scheduling by mandating advance notice and predictable hours. 

These initiatives offer promising templates but also highlight implementation challenges. Businesses may face increased compliance costs, particularly small or resource-constrained employers who rely on just-in-time scheduling. Similarly, calls for algorithmic transparency may encounter pushback due to proprietary concerns or technical opacity. On the consumer side, education campaigns must contend with low engagement, information fatigue, and resistance to behavior change. Therefore, any effective intervention must balance regulatory ambition with practical feasibility, incorporating incentives, phased rollouts, and stakeholder input to ensure both compliance and impact. Future work could formalize these trade-offs and assess policy outcomes through simulations or field studies.

We contend that consumer agency was limited even before the prevalent use of algorithms, but automation has made existing issues more widespread. Thus, to motivate our interventions and through the proposed interventions, we aim to support gradual improvements to empower consumers,  but they may not fully restore agency in all contexts.

That is, the specificity, effectiveness, and desirability of certain interventions are very context-dependent. Therefore, our proposals are intended as part of an ongoing, necessary conversation about how to intervene in algorithmic systems that increasingly shape consumption. The interventions' aim was not to present a specific \textit{finished} solution but to reframe the debate around actionable steps given the algorithmic erosion of agency. Interventions create friction in exploitative processes and empower consumers within those constraints. Relevant contributions also add value to current debates on responsible AI and digital governance.

Nonetheless, the intervention we formally analyze in {\S\ref{sec:main-intervention}} is very specific and focuses on a more concrete objective; namely, encouraging effective budgeting and value-based deliberation in consumption decisions.

\section{Conclusion}
\label{sec:conclusions}

This paper's core contribution is a formal analysis of financial instability arising from diminished consumption agency. Our framework tackles three key questions relevant to consumption agency in today’s economy: \begin{itemize}
  \item What are the aftereffects of a lack of agency in consumption, particularly when exacerbated by algorithms? How do these effects manifest across different scenarios, such as obligatory consumption, impulsive algorithm-driven purchases, and dynamic work scheduling? \item Why should consumption be considered a value that needs to be fostered rather than a given in the modern economy? \item What are the effects of educational and policy regulatory interventions served as consumption value deliberation methods (by enabling consumers to achieve true agency and empowering them in their decision-making)?
\end{itemize}

We answer these questions through formal analyses of intertemporal consumption models with limited consumption agency and illustrate the resulting adverse effects (like early financial ruin and diminished consumer utility) that arise when consumer agency is eroded by systemic barriers, algorithmic manipulation, and obligatory consumption patterns. We also illustrate the positive effects of recognizing consumer agency as a value and adopting value deliberation as a mindset for interventions at both individual and societal levels. Potential measures include regulatory oversight to mitigate manipulative practices and consumer education initiatives to improve financial literacy.

\bibliographystyle{unsrt}  
\bibliography{references}

\cleardoublepage 
\appendix
\cleardoublepage
\section{True Agency in an Ideal Setting}\label{sec:agency}

In an ideal world, where individuals enjoy true agency, they would be capable of avoiding financial ruin indefinitely. We formally state this below,

\vspace{0.2cm}

\begin{theorem}\label{thm:agency-thm}
Assuming $Y_t$ ($Y_t \ge 0$) is drawn from a distribution with known mean (denoted by $Y$), if we allow the agent to vary the consumption without any variability constraints other than $0\le c_t \le a_t$, then under isoelastic utility with $\lambda > 0$ (and $\lambda \neq 1$),  
a rational agent with utility maximization goals could consume with infinite horizon without going to ruin early. 
\end{theorem}

\vspace{0.2cm}

The proof of the Theorem~\ref{thm:agency-thm} is as follows,

\vspace{0.2cm}

\begin{proof}
Assume you are given parameter $\beta$ and utility function $u(c)$. Consider the optimal consumption sequence that you get where the agent goes to ruin at time point $T$ (here, asset ruin is assumed to be the point where the available assets for the next iteration reach $0$). Let this be $C = \{c_{1},c_{2},\dots,c_{T}\}$. This is possible only if $Y_{T+1} = 0$ since $c_T \le a_T$. 

Now consider the amended sequence where $C^* = \{c_{1},c_{2},\dots,c_{T} - \epsilon, \epsilon, \dots\}$. We can see that this does not go to asset ruin at point $T$ since we consume less than $c_T$ and we still have $\epsilon$ left. 

We can see that the total utility for $C$ is $\sum_{t=1}^{T-1}\beta^t u(c_t)+\beta^T \cdot u(c_T)$ and we can also see that the total utility for $C^*$ is $>\sum_{t=1}^{T-1}\beta^t u(c_t)+\beta^T \cdot u(c_T-\epsilon)+\beta^{T+1} \cdot u (\epsilon)$. Therefore, the change in utility is,
\begin{align*}
\delta_u &> \left(\sum_{t=1}^{T-1}\beta^t u(c_t)+\beta^T \cdot u(c_T-\epsilon)+\beta^{T+1} \cdot u (\epsilon)\right)\\
&-\left(\sum_{t=1}^{T-1}\beta^t u(c_t)+\beta^T \cdot u(c_T)\right)\\
&> \beta^T \cdot u(c_T-\epsilon)+\beta^{T+1} \cdot u (\epsilon)-\beta^T \cdot u(c_T)
\end{align*}

Now what we are left to do is to argue that there exists an $\epsilon>0$ that satisfies $\beta^T \cdot u(c_T- \epsilon)+\beta^{T+1} \cdot u (\epsilon) \ge \beta^T \cdot u(c_T) \implies u(c_T- \epsilon)+\beta \cdot  u(\epsilon) \ge u(c_T)$. Let $g(\epsilon) = u(c_T- \epsilon)+\beta \cdot  u(\epsilon)$, then we would like to claim that there exists a $\epsilon>0$ such that $g(\epsilon) \ge g(0)$. Consider the $\epsilon$ that maximizes $g(.)$, then for this $\beta u'(\epsilon) = u'(c_T-\epsilon)$. 

\vspace{0.2cm}

\textbf{Isoelastic utility:} Let $u(c)=\frac{c^{1-\lambda}-1}{1-\lambda}$. Assume $\lambda \neq 1$. Let $\epsilon >0 $ be such that $\beta \epsilon^{-\lambda} = (c_T-\epsilon)^{-\lambda} \implies \beta u'(\epsilon) = u'(c_T-\epsilon)$. For notational convinience, let $s=\beta \epsilon^{-\lambda} = (c_T-\epsilon)^{-\lambda}$ Then, we can see that,
\begin{align*}
g(\epsilon) &= \frac{(c_T-\epsilon)^{1-\lambda}-1}{1-\lambda} + \beta \frac{(\epsilon)^{1-\lambda}-1}{1-\lambda} \\
& = \frac{(c_T-\epsilon)}{1-\lambda} (c_T-\epsilon)^{-\lambda} + \beta (\epsilon)^{-\lambda}\frac{(\epsilon)}{1-\lambda}-\frac{1+\beta}{1-\lambda}\\
&= \frac{c_T}{1-\lambda} s -\frac{1+\beta}{1-\lambda}\\
& = \frac{c_T^{1-\lambda}}{1-\lambda} \frac{c_T^\lambda}{(c_T-\epsilon)^\lambda} -\frac{1+\beta}{1-\lambda} \\
&= \frac{c_T^{1-\lambda}}{1-\lambda} \frac{c_T}{(c_T-\epsilon)}^\lambda -\frac{1+\beta}{1-\lambda}\\
& = \frac{c_T^{1-\lambda}}{1-\lambda} \left(1+\frac{\epsilon}{(c_T-\epsilon)}\right)^\lambda -\frac{1+\beta}{1-\lambda}\\
&= \frac{c_T^{1-\lambda}}{1-\lambda} \left(1+\beta^{1/\lambda}\right)^\lambda -\frac{1+\beta}{1-\lambda}\\
&\ge (1+\beta) \frac{c_T^{1-\lambda}-1}{1-\lambda} \ge g(0)
\end{align*}
as desired.
\end{proof}

\section{Work Scheduling: Formal Theorem and the Corresponding Proofs (Restated from \texorpdfstring{\cite{nokhiz2024modeling,nokhiz2025counting,nokhiz2025counting1}}{109,110,111})}\label{sec:appendix-work-schedule}

Let $ y_t $, $ x_t $, and $ c_t $ represent income, assets, and consumption at time $ t $, respectively, with $ T $ indicating the job duration and $ u(c) = \sqrt{c} $ representing the utility derived from consumption $ c $. Also, assume no discounting ($ \beta = 1 $) and that income $ y_t $ lies within the range $ [0, Y] $. Previous research in ~\cite{nokhiz2024modeling,nokhiz2025counting,nokhiz2025counting1} shows that in certain cases, a worker with lookahead privileges has a distinct advantage over one without, and this benefit increases linearly with the amount of lookahead.

Before presenting the main theorem, \cite{nokhiz2024modeling,nokhiz2025counting,nokhiz2025counting1} introduces the following lemma:  
\vspace{0.2cm}

\begin{lemma}\label{lem:lb-helper}
Let $ a \in (1/2, 1) $ be a constant, and $ w \in (0,1) $. Then:  
\[ \sqrt{w} \le \sqrt{a} + \frac{1}{2\sqrt{a}} (w-a) - \frac{1}{8} (w-a)^2. \]  
\end{lemma}

\vspace{0.2cm}

\begin{proof}
We have:
\begin{align*}
&\sqrt{w} - \sqrt{a} - \frac{1}{2\sqrt{a}} (w-a) \\
&= (w-a) \left( \frac{1}{\sqrt{w}+\sqrt{a}} - \frac{1}{2\sqrt{a}} \right) \\
&= \frac{ (w-a)(a-w) }{2\sqrt{a}(\sqrt{w}+\sqrt{a})^2} \\
&= - \frac{ (w-a)^2 }{2\sqrt{a}(\sqrt{w}+\sqrt{a})^2}.
\end{align*}
Since the denominator satisfies $ \leq 8 $, we get $\sqrt{w} \le \sqrt{a} + \frac{1}{2\sqrt{a}} (w-a) - \frac{1}{8} (w-a)^2$.
\end{proof}

\vspace{0.2cm}

\begin{theorem}\label{thm:main-lookahead-formal}
Consider two individuals: one with a lookahead of $ k $ steps and the other without any lookahead. Let $ c_1, c_2, \dots, c_T $ denote the consumption of the individual with $ k $-step lookahead, and $ z_1, z_2, \dots, z_T $ denote the consumption of the individual without lookahead. There exist instances where all incomes lie within the range $ [0, Y] $, such that:
\[
\sum_{t=1}^T \sqrt{c_t} - \sum_{t=1}^T \sqrt{z_t} \geq \Omega(k \sqrt{Y}).
\]
\end{theorem}

\begin{proof}

\cite{nokhiz2024modeling,nokhiz2025counting,nokhiz2025counting1} assumes that both individuals start with $ a_1 = 0 $. Also, since incomes can be scaled without loss of generality, the assumption is that $Y=1$. Let $y_t$ be defined as:
\[
y_t = 
\begin{cases} 
1 & \text{for } t \leq k/2, \\
x & \text{for } k/2 < t \leq k,
\end{cases}
\]
where $ x $ is uniformly sampled from the interval $[0,1]$. Both individuals are aware of this input distribution. 

For simplicity, they assume that the total time horizon $ T $ equals $ k $. This assumption can be relaxed by setting $ y_t = \frac{(1+x)}{2}$ for all $ t > k $. 

Firstly, they examine the individual using $k$-step lookahead. Since this individual can observe the value of $x$, they opt to consume an amount of $\frac{1+x}{2}$ at each time step. This strategy is viable because the income during the first $k/2$ steps is $y_t = 1$, ensuring their assets are adequate to sustain this consumption. As a result, the total utility (given that $T = k$) amounts to $k \sqrt{\frac{1+x}{2}}$.

Next, they examine the individual without lookahead. Intuitively, this person is unable to predict the value of $x$, making it impossible to consistently consume $ \frac{1+x}{2} $. \cite{nokhiz2024modeling,nokhiz2025counting,nokhiz2025counting1} shows that even with complex strategies, the individual without lookahead cannot achieve a high total utility.

They apply Yao's min-max theorem~\cite{motwani-raghavan}, which asserts that for a specific input distribution, the best algorithm for that distribution is deterministic. Thus, to derive the required lower bound, it suffices to examine deterministic algorithms and evaluate the expected difference in total utility across the random choice of $x$. For any deterministic algorithm, during the time steps $ 1, \dots, (k/2) $, the algorithm observes an income of $ 1 $, meaning the consumption values $ z_1, z_2, \dots, z_{(k/2)} $ are predetermined. Let $ S = z_1 + z_2 + \dots + z_{(k/2)} $.

Consider the scenario where the following condition holds:
\begin{equation}\label{eq:lb-good-condition}
\left| S - \frac{k}{2} \cdot \frac{1+x}{2} \right| > c \cdot k,
\end{equation}
for some parameter $ c $. In this case, \cite{nokhiz2024modeling,nokhiz2025counting,nokhiz2025counting1} demonstrates that $ \sum_i \sqrt{z_i} $ is significantly less than $ k \sqrt{\frac{1+x}{2}} $.

Let $ a = \frac{1+x}{2} $. From the assumption, it is known that $ | z_1 + z_2 + \dots + z_{k/2} - \frac{k}{2} a| > ck $, which implies 
\begin{equation}\label{eq:lb-variation} 
\sum_{i \le k} |z_i - a| > ck.
\end{equation}
Using Lemma~\ref{lem:lb-helper}:
\[  
\sum_{i \le k} \sqrt{z_i} \le k \sqrt{a} + \frac{1}{2\sqrt{a}} \sum_{i \le k} (z_i -a) - \frac{1}{8} (z_i - a)^2. 
\]
Since the total consumption cannot exceed the total income (equal to $ ka $), the middle term on the right is $ \le 0 $. This simplifies to:
\[ 
\sum_{i \le k} \sqrt{z_i} \le k \sqrt{a}  - \frac{1}{8} (z_i - a)^2. 
\]
Using the Cauchy-Schwarz inequality and~\eqref{eq:lb-variation}:
\[ 
\sum_i (z_i - a)^2 \ge \frac{1}{k} \left( \sum_{i} |z_i -a| \right)^2 > c^2 k.
\]
Combining these inequalities, it follows that:
\[ 
\sum_i \sqrt{z_i} \le k\sqrt{a} - \frac{c^2 k}{8}.
\]
This result demonstrates that if the deterministic algorithm chooses $z_i$ values that satisfy \eqref{eq:lb-good-condition}, the no-lookahead algorithm results in a total utility that is $ \Omega(k) $ worse than that of the lookahead algorithm.

To complete the proof, \cite{nokhiz2024modeling,nokhiz2025counting,nokhiz2025counting1} shows that if $ x $ is chosen uniformly at random from $ (0,1) $, the condition~\eqref{eq:lb-good-condition} holds with a constant probability for some $ c > 0 $. Since $ S $ is fixed, then:
\[ 
\left| \frac{2S}{k} - \frac{1+x}{2} \right| > 2c, 
\]
or equivalently:
\[ 
\left| \frac{4S}{k} - 1 - x \right| > 4c. 
\]
For any fixed $ \alpha $, if $ x \sim_{\text{uar}} (0,1) $, the probability that $ |\alpha - x| \le 1/3 $ is at most $ 2/3 $. Therefore, the condition holds with $ c = 1/12 $ and with a probability of at least $ 1/3 $.

In summary, with a probability of $1/3$, the no-lookahead algorithm performs $\Omega(k)$ worse than the lookahead algorithm and cannot achieve better performance. Therefore, the \emph{expected} utility gap is also $\Omega(k)$. By applying Yao's min-max theorem, this lower bound holds for any (potentially randomized) algorithm as well.
\end{proof}

\section{Modeling Debt and Behavioral Anomalies}
\label{sec:app-debttt}

An important direction for future work lies in extending our model to incorporate consumer debt and behavioral deviations from standard discounted utility theory, such as hyperbolic discounting. These features play a central role in real-world consumption decisions and would enrich the theoretical analysis of agency under financial constraint.

Debt, particularly in the form of credit cards, Buy Now Pay Later (BNPL) schemes, and payday loans, serves as a mechanism for smoothing consumption over time when current income or assets are insufficient. In classical intertemporal models, allowing for debt would relax the constraint $c_t \leq a_t $
  and introduce borrowing constraints or interest-dependent repayment obligations. However, this also introduces new risks: under limited agency (e.g., impulsive purchases or unpredictable income), access to debt can amplify financial fragility, increase the likelihood of ruin, and create feedback loops where short-term choices impose long-term burdens. For example, manipulated urgency or misleading repayment terms can cause consumers to overborrow without fully anticipating the long-run costs.

Behavioral anomalies, such as hyperbolic discounting {\cite{ainslie1992hyperbolic, frederick2002time}}, refer to empirically observed patterns where individuals heavily overweight immediate rewards compared to future outcomes, far more than exponential discounting predicts. Formally, hyperbolic discounting modifies the temporal weight on future utility such that the discount factor becomes time-inconsistent: individuals may plan to save in the future but fail to follow through when the future becomes the present. This deviation leads to dynamic inconsistency, myopic spending, and a stronger susceptibility to manipulative design features like flash sales, scarcity cues, or short-term promotions. In the presence of algorithmic persuasion, hyperbolic discounting can exacerbate consumer vulnerability by aligning system-level nudges with consumers’ inherent present bias. 

While both debt modeling and hyperbolic discounting are highly relevant to the broader discourse on consumption agency, we deliberately chose not to include them in our current formalism for tractability and interpretability. Our primary objective was to isolate and analyze the structural impact of limited agency (particularly under algorithmic and systemic constraints) within a clean, utility-maximizing framework. Including debt dynamics or non-standard temporal discounting would add additional modeling complexity; nonetheless, these extensions are conceptually aligned with our goals and represent promising avenues for further study.

\vspace{0.5cm}

\section{Additional Experimental Analysis}
\label{sec:app-all-exp}
In this section, using our model, we explore four different sets of experiments regarding the algorithmic persuasion and ruin, behaviors of different income classes, behaviors of different individuals based on education level (to avoid social uniformity and explore more diverse subgroups that could potentially be underprivileged), and work scheduling effects on financial utility.

\subsection{Experimental Setup}\label{sec:exp-setup}

We begin by describing the key components of our experimental setup.

\textbf{Model of Computation:} Our experiments are based on the computational framework introduced in \S{\ref{sec:main-model}}. This model employs value iteration to compute optimal consumption paths, which are then used in the simulations to determine final consumption and asset trajectories. The discounting factor $\beta$ is set to 0.95 based on what is commonly practiced in literature in general settings \cite{patnaik2022role, cooper2014discounting}.

\textbf{Assets and Income:} The initial asset level is set at \$141{,}140, corresponding to the median net worth reported by the Federal Reserve in 2019~{\cite{fedres}}. Monthly income is drawn from a log-normal distribution with a mean of \$5{,}957.25 and a standard deviation of \$378.74, derived from annual income statistics provided by the U.S. Bureau of Labor Statistics for 2019~{\cite{minsub}}.

To explore the impact of income level, we also simulate both low- and high-income scenarios in {\S\ref{sec:app-income-expt}}. We consider the individuals with income in the range from \$1,250 to \$2,500 as the low-income individuals and the individuals with income in the range from \$8,334 to \$12,500 as the high-income individuals {\cite{minsub}}. The low-income group receives income from a distribution with a mean of \$1{,}899.33 and standard deviation of \$77, while the high-income group’s distribution has a mean of \$8{,}869.92 and standard deviation of \$199.60~{\cite{minsub}}.

\textbf{Simulation Setup and Plots:} Each experiment simulates 50{,}000 agents, with each agent planning over a 100-month horizon. The point of ruin is defined as the first month in which an agent’s asset level falls below zero. From these simulations, we generate histograms representing the distribution of ruin times across agents.

\vspace{0.2cm}

\subsection{Algorithmic Persuasion, Impulsive Consumption and Ruin}

In this experiment, we investigate how algorithmic persuasion and impulsive consumption, under minimum subsistence conditions, contribute to financial ruin. The model and experimental parameters follow the standard configuration detailed in \S{\ref{sec:exp-setup}} with the general income/asset setting for the entire underlying population with no income classes defined yet.  The mean expenditure value is \$5{,}253 per month based on the statistics provided by the U.S. Bureau of Labor Statistics for 2019~{\cite{minsub}}. Figure~{\ref{fig:ruin-algorithm}} presents the distribution of ruin times for the algorithmic persuasion and impulsive consumption scenario.

\begin{figure}[ht]
    \centering
\includegraphics[width=0.4\columnwidth]{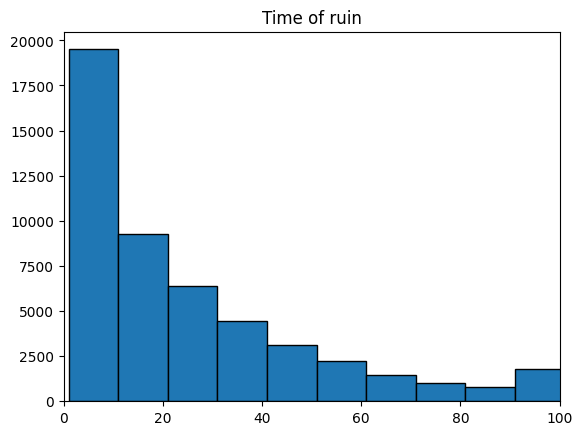}
    \caption{The figure shows the distribution of ruin times under algorithmic persuasion and impulsive consumption. The experiment simulates 50,000 agents, and the resulting distribution reflects the time at which each agent experiences ruin. The $x$-axis denotes the time step (month) when ruin occurs, while the $y$-axis indicates the number of agents who go to ruin at each time step.}
    \label{fig:ruin-algorithm}
\end{figure}

\textbf{Analysis:} Figure~{\ref{fig:ruin-algorithm}} illustrates that agents operating under minimum subsistence constraints tend to experience ruin within a short number of time steps. A significant fraction of agents reach ruin within the first 10 steps, highlighting the heightened financial vulnerability caused by impulsive consumption in such conditions. This rapid onset of ruin underscores the short-term instability faced by these agents. The frequency of ruin decreases steeply over time, indicating that as the simulation progresses, fewer agents manage to maintain solvency. Nonetheless, under the standard setup, a subset of agents is able to delay ruin significantly, with some surviving until the end of the planning horizon.

\subsection{Diverse Subgroups: Low Income vs. High Income}
\label{sec:app-income-expt}
In this experiment, we examine how income influences the likelihood of financial ruin. The motivation is to consider analyzing varied demographic segments and striving to prevent assumptions of social homogeneity. 

We consider agents with low-income and agents with high-income as defined in \S{\ref{sec:exp-setup}}. For these income-specific scenarios, the mean expenditure is \$2{,}850 for the low-income group and \$7{,}082.83 for the high-income group, following the same 2019 statistics~{\cite{minsub}}. Other settings are similar to \S{\ref{sec:exp-setup}}.

Our computational behavioral model incorporates a discount factor $\beta$, which represents the consumer's level of patience. While in general settings, more commonly used values like 0.95 can be used {\cite{patnaik2022role,cooper2014discounting},  this parameter has been studied in more depth. Findings suggest that $\beta$ can vary with income level. Thus, for individuals with a low income we set a lower discount factor (around $\beta = 0.5$), whereas those with a high income we set a higher discount factor (around $\beta = 0.9$) per prior research on $\beta$ for these income classes~{\cite{ericson2019intertemporal}}. In Figure~{\ref{fig:ruin-income}} we provide the distributions of the time of ruin for the low-income and high-income individuals.

\vspace{0.4cm}
 
\begin{figure}[ht]
    \centering

    \begin{minipage}{0.45\textwidth}
        \centering
        \includegraphics[width=0.95\textwidth]{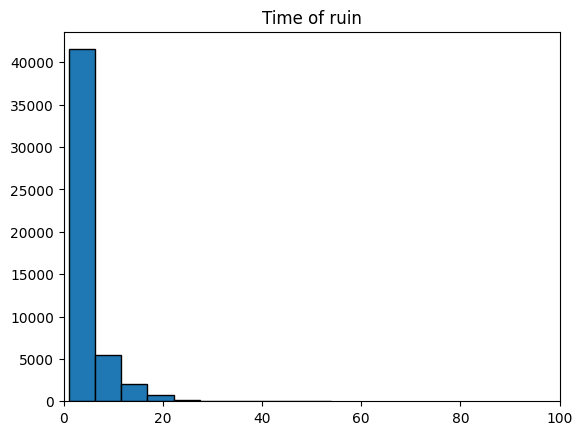} \\
        \small High-income
    \end{minipage}
    \hfill
    \begin{minipage}{0.45\textwidth}
        \centering
        \includegraphics[width=0.95\textwidth]{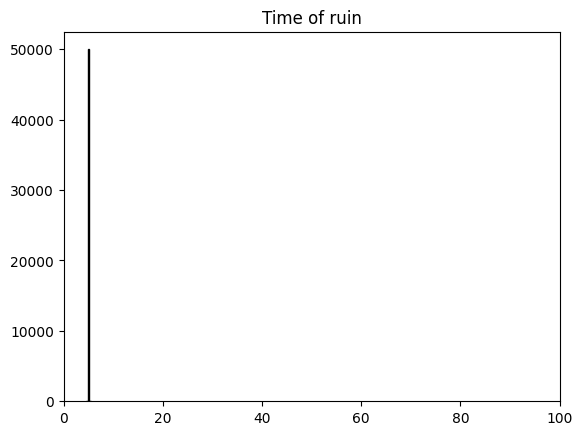} \\
        \small Low-income
    \end{minipage}

    \caption{The figure shows the distributions of ruin times for low income households and high income households. The experiment simulates 50,000 agents, and the resulting distribution reflects the time at which each agent experiences ruin. The $x$-axis denotes the time step (month) when ruin occurs, while the $y$-axis indicates the number of agents who go to ruin at each time step.}
    \label{fig:ruin-income}
\end{figure}

\textbf{Analysis:} Figure~{\ref{fig:ruin-income}} illustrates the distribution of time to ruin under impulsive consumption for agents with low and high income levels. In the low-income scenario, nearly all 50,000 agents go to ruin immediately within the first time step. The histogram is sharply concentrated at the origin with no meaningful tail, indicating extreme financial fragility. This suggests that under low income and impulsive behavior, agents are incapable of sustaining even minimal subsistence, leading to near-instantaneous collapse. 

In contrast, while many high-income agents also experience early ruin, the distribution shows noticeably more spread, with a non-trivial number of agents surviving for longer durations. This reflects the buffering effect of higher income: though impulsive consumption still leads most high-income agents to financial ruin relatively quickly, greater income provides some capacity to delay collapse. Overall, the results highlight that income level influences time to ruin. Low-income agents exhibit much lower financial resilience, while high-income agents display a modest degree of resilience.

\subsection{Diverse Subgroups: High School vs. College Graduates}

In this experiment, we examine another set of experiments to explore the effects of impulsive consumption for different subgroups. Specifically, we analyze how educational attainment affects financial stability based on the effects it has on $\beta$ in consumption decisions (and how these differences can lead agents to experience ruin).

Once again, the discount factor $\beta$ has been studied in the prior literature, and findings suggest that $\beta$ may vary with education. In particular, individuals with a high school education typically exhibit a lower discount factor (around $\beta = 0.5$), whereas those with a college degree tend to have a higher value (around $\beta = 0.83$)~{\cite{ericson2019intertemporal}}. As shown in prior research {\cite{nokhiz2025counting}}, people with higher education tend to have higher income, more financial flexibility, and a resultant financial utility, which indicates that the lower-education group is the underprivileged group in this setting.

The experiment setup is similar to the general setting in \S{\ref{sec:exp-setup}}. In Figure~{\ref{fig:ruin-education}} we provide the distributions of the time of ruin for the high-school graduates and college graduates.

\begin{figure}[ht]
    \centering

    \begin{minipage}{0.45\textwidth}
        \centering
        \includegraphics[width=0.95\textwidth]{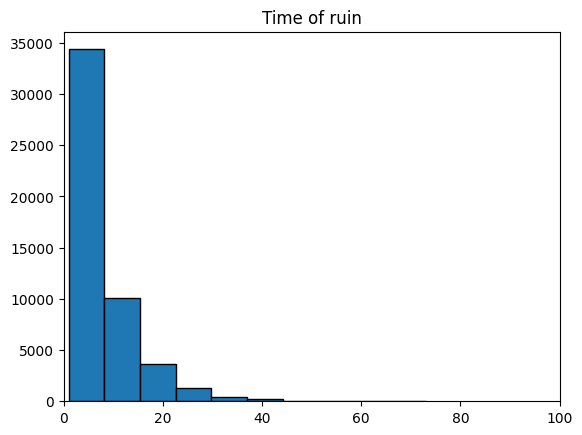} \\
        \small College degree
    \end{minipage}
    \hfill
    \begin{minipage}{0.45\textwidth}
        \centering
        \includegraphics[width=0.95\textwidth]{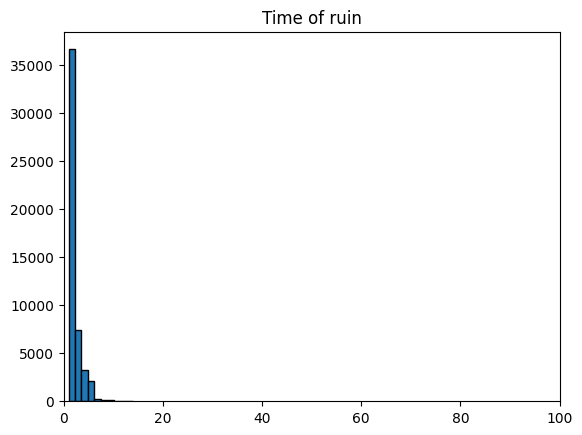} \\
        \small High-school diploma
    \end{minipage}

    \caption{The figure shows the distributions of ruin times for college degree holders and high school diploma holders. The experiment simulates 50,000 agents, and the resulting distribution reflects the time at which each agent experiences ruin. The $x$-axis denotes the time step (month) when ruin occurs, while the $y$-axis indicates the number of agents who go to ruin at each time step.}
    \label{fig:ruin-education}
\end{figure}

\textbf{Analysis:} Figure~{\ref{fig:ruin-education}} illustrates the distribution of time to ruin under impulsive consumption for agents with high school diploma and college degree. In the high school diploma scenario, nearly all 50,000 agents go to ruin very quickly within the first 20 time steps. The histogram indicates financial fragility. This suggests that underprivileged groups (here with lower $\beta$ in the impulsive case) are incapable of efficiently sustaining themselves, leading to sudden collapse. 

In contrast, while many college degree-holding agents also experience early ruin, the distribution shows significant spread, with a significant number of agents surviving for longer durations. This reflects the differences that consumption bounds could potentially have on diverse education level subgroups. 

\vspace{0.5cm}

Overall, the results highlight that various factors representing different subpopulations (like $\beta$ here for group differentiation to avoid social uniformity) can impact time to ruin.

\subsection{Work Schedule Instability Analysis}

Prior work \cite{nokhiz2024modeling,nokhiz2025counting, nokhiz2025counting1} show both empirically and formally that underprivileged workers (particularly those with limited foresight into their future schedules) experience significantly lower financial utility compared to those with greater lookahead. The ability to anticipate upcoming work hours allows individuals to better plan consumption and savings, ultimately resulting in higher utility. Additionally, the effect of income is straightforward: higher earnings allow individuals to meet consumption needs more comfortably due to their higher utility scale across all lookahead conditions. Moreover, disadvantaged groups holding more rapidly depreciating assets have an even greater need for advance information about their work schedules, given their already vulnerable financial standing.

These studies also reveal that employees who are given advanced notice of schedule changes based on regulations like the Schedules that Work Act of 2014 (H.R. 5159) and Retail Workers Bill of Rights in San Francisco {\cite{golden2015irregular}} experience improved financial well-being. Advanced planning allows workers with variable schedules to gain greater autonomy over their financial planning. That is, they can manage their consumption and savings more effectively, which leads to better financial decision-making and agency.

\end{document}